\documentclass[a4paper,UKenglish,cleveref,autoref,thm-restate]{lipics-v2021}

\hideLIPIcs
\usepackage{microtype}
\usepackage{graphicx}
\usepackage{tikz}
\usetikzlibrary{fit,positioning,shapes,calc,decorations.pathmorphing,decorations.pathreplacing,calligraphy,backgrounds}

\usepackage{bbm,amsmath,amssymb,amsfonts}
\usepackage{enumitem}
\usepackage{array}
\usepackage{pifont}
\newcolumntype{L}[1]{>{\raggedright\let\newline\\\arraybackslash\hspace{0pt}}m{#1}}
\newcolumntype{C}[1]{>{\centering\let\newline\\\arraybackslash\hspace{0pt}}m{#1}}
\newcolumntype{R}[1]{>{\raggedleft\let\newline\\\arraybackslash\hspace{0pt}}m{#1}}

\usepackage{boxedminipage}
\usepackage{ifthen}
\newcommand{\problemdef}[4][XXXEMPTYLABELXXX]{%
	\begin{center}
		\begin{boxedminipage}{\textwidth}
			\textsc{{#2}}\ifthenelse{\equal{#1}{XXXEMPTYLABELXXX}}{}{\label{#1}}\\[2pt]
			\renewcommand{\arrayrulewidth}{0pt}
			\begin{tabular}{@{\hspace{0.007\textwidth}}r@{\hspace{0.007\textwidth}}p{0.87\textwidth}@{\hspace{0.007\textwidth}}}
				\textit{Input:}  & {#3}\\
				\textit{Output:} & {#4}
			\end{tabular}
		\end{boxedminipage}
	\end{center}
}

\title{Robust Algorithms for Path and Cycle Problems in Geometric Intersection Graphs}

\author{Malory Marin}{ENS de Lyon, CNRS, Université Claude Bernard Lyon 1, LIP, UMR 5668, 69342, Lyon cedex 07, France}{malory.marin@ens-lyon.fr}{}{}
\author{Jean-Florent Raymond}{CNRS, ENS de Lyon, Université Claude Bernard Lyon 1, LIP, UMR 5668, 69342, Lyon cedex 07, France}{jean-florent.raymond@cnrs.fr}{}{}
\author{Rémi Watrigant}{Université Claude Bernard Lyon 1, CNRS, ENS de Lyon, LIP, UMR 5668, 69342, Lyon cedex 07, France}{remi.watrigant@univ-lyon1.fr}{}{}
\authorrunning{M.\ Marin, J.-F.\ Raymond and R.\ Watrigant} 
\titlerunning{Robust Algorithms for Path and Cycle Problems in Geometric Intersection Graphs} 

\Copyright{Malory Marin, Jean-Florent Raymond and Rémi Watrigant}


\nolinenumbers 

\begin{document}
\maketitle%

\begin{abstract}
We study the design of robust subexponential algorithms for classical connectivity problems on intersection graphs of similarly sized fat objects in $\mathbb{R}^d$. In this setting, each vertex corresponds to a geometric object, and two vertices are adjacent if and only if their objects intersect. We introduce a new tool for designing such algorithms, which we call a \emph{$\lambda$-linked partition}. This is a partition of the vertex set into groups of highly connected vertices. Crucially, such a partition can be computed in polynomial time and does not require access to the geometric representation of the graph.

We apply this framework to problems related to paths and cycles in graphs. First, we obtain the first \emph{robust} \emph{ETH-tight} algorithms for \textsc{Hamiltonian Path} and \textsc{Hamiltonian Cycle}, running in time $2^{O(n^{1-1/d})}$ on intersection graphs of similarly sized fat objects in $\mathbb{R}^d$. This resolves an open problem of de~Berg \emph{et~al.} [STOC~2018] and completes the study of these problems on geometric intersection graphs from the viewpoint of ETH-tight exact algorithms.

We further extend our approach to the parameterized setting and design the first robust subexponential parameterized algorithm for \textsc{Long Path} in any fixed dimension~$d$. More precisely, we obtain a randomized robust algorithm running in time $2^{O(k^{1-1/d}\log^2 k)}\, n^{O(1)}$ on intersection graphs of similarly sized fat objects in $\mathbb{R}^d$, where $k$ is the natural parameter. Besides $\lambda$-linked partitions, our algorithm also relies on a low-treewidth pattern covering theorem that we establish for geometric intersection graphs, which may be viewed as a refinement of a result of Marx-Pilipczuk~[ESA~2017]. This structural result may be of independent interest.

\end{abstract}

\section{Introduction}

Given a set $\mathcal{F}$ of objects in $\mathbb{R}^d$, the \emph{intersection graph} of $\mathcal{F}$ is defined as the graph having one vertex for each object in $\mathcal{F}$, and an edge between two vertices whenever the corresponding objects intersect. One of the most extensively studied classes of intersection graphs are the \emph{unit disk graphs}, obtained when the objects are disks in $\mathbb{R}^2$ of identical radius. In this work, we consider families $\mathcal{F}$ of \emph{similarly sized $\beta$-fat} objects in $\mathbb{R}^d$ (for some constants $d,\beta\geq 1$) which means that for every object $O \in \mathcal{F}$, there are two balls $B_{\mathrm{in}}$ and $B_{\mathrm{out}}$ of $\mathbb{R}^d$ such that $B_{\mathrm{in}} \subseteq O \subseteq B_{\mathrm{out}}$,
where $B_{\mathrm{in}}$ has radius~$1$ and $B_{\mathrm{out}}$ has radius~$\beta$.

In their seminal work, de~Berg \emph{et~al.}~\cite{de2018framework} introduced a general framework for deriving algorithms with running times that are tight under the Exponential Time Hypothesis (ETH) for a broad range of classical problems on intersection graphs of similarly sized fat objects. Their framework encompasses problems such as \textsc{Maximum Independent Set}, \textsc{Dominating Set}, \textsc{Steiner Tree}, and \textsc{Hamiltonian Cycle}. The key idea of their approach is the construction of a partition $\mathcal{P} = (V_1, \ldots, V_t)$ of the vertex set of the input graph $G$ satisfying the following properties:
\begin{enumerate}[label = (\roman*)]
    \item each induced subgraph $G[V_i]$ can be further partitioned into at most $\kappa$ cliques, and 
    \item the quotient graph of the parts, denoted $G_{\mathcal{P}}$, has bounded degree and treewidth $O(n^{1 - 1/d})$.
\end{enumerate}
Notice that $\mathcal{P}$ can be computed \emph{without} access to the geometric representation of $F$, and is enough to solve all the aforementioned problems except for \textsc{Hamiltonian Cycle}. Algorithms of this type, which operate solely on the intersection graph and do not require a geometric representation, are referred to as \emph{robust} algorithms. We stress here that robustness is a substantial advantage since for many classes of intersection graphs such as unit disk graphs, computing a representation is $\exists \mathbb{R}$-complete~\cite{kang2011sphere}.

However, the algorithm of \cite{de2018framework} solving \textsc{Hamiltonian Cycle} requires one additional step: computing a partition of each~$V_i$ into a bounded number of cliques.
For this step, de Berg \emph{et~al.} 
used the geometric representation. They explicitly left open the question of whether a robust ETH-tight algorithm for \textsc{Hamiltonian Cycle} could be obtained. One possible direction toward such a result would be to design a robust algorithm that computes a partition into a constant number of cliques whenever the input graph is known to contain one. Unfortunately, recent advances on clique partitions of geometric graphs suggest that this task is likely to be difficult~\cite{koana2024subexponential}.\footnote{There is no $2^{o(n)}$-time algorithm for this problem in unit ball graphs of $\mathbb{R}^5$, unless the ETH fails \cite{koana2024subexponential}.} It is worth noting that, in the case of unit disk graphs, such a clique partition can be obtained using the approach of~\cite{pirwani2012weakly}.

We circumvent this problem by using a new type of partition into highly connected subgraphs called \emph{$\lambda$-linked partitions}. We show the two following crucial properties : these partitions can be used to solve the considered problems (instead of partitions into cliques) and they can be computed in polynomial time (and, importantly, without relying on the geometric representation) in graphs known to admit a partition into a constant number of cliques.

Combining these findings with the framework of de Berg \emph{et~al.}, we obtain the following.

\begin{theorem}\label{thm:main1}
For every constants $d\geqslant 1$ and $\beta\geqslant 1$ there is a robust algorithm solving \textsc{Hamiltonian Path} (resp.\  \textsc{Hamiltonian Cycle}) in time $2^{O\left (n^{1 - 1/d} \right )}$ on intersection graphs of similarly sized $\beta$-fat objects in~$\mathbb{R}^d$.
\end{theorem}

This result matches the ETH-based lower bound of \cite{de2018framework} and resolves an open problem from the same paper.
\medskip

To further demonstrate the applicability of our techniques, we show in a second part of the paper how they can be used to deal with parameterized version of \textsc{Hamiltonian Path}, that is, \textsc{Long Path} parameterized by the length of the path. 
For this problem and for \textsc{Long Cycle}, Fomin \emph{et~al.} proved in \cite{fomin2020eth} that there are algorithms running in time $2^{O(\sqrt{k})} \cdot n^{O(1)}$ on intersection graphs of similarly sized fat objects in~$\mathbb{R}^2$.
Their approach relies on computing a \emph{clique partition} of the vertex set and analyzing the resulting \emph{clique-grid graph}. The idea of using clique partitions in the context of subexponential parameterized algorithms was introduced in~\cite{fomin2019finding} and has since become a central tool in the area~\cite{xue2025parameterized}. However, the clique-grid graph is typically defined using a geometric representation of the input graph. Designing a robust subexponential parameterized algorithm for \textsc{Long Path} in dimensions $d \geqslant 3$ has remained an open problem. Indeed, while the bidimensionality-based approach of~\cite{fomin2020eth} does not generalize to higher dimensions, the high-dimensional approach of~\cite{fomin2019finding}, based on Baker’s technique~\cite{erlebach2005polynomial}, again requires access to the geometric representation.

We show that this dependency can be removed by combining two key techniques. The first consists, as in the case of the \textsc{Hamiltonian Path} problem, in replacing the clique partition with a partition into highly connected parts. The second is an adaptation of the \emph{low-treewidth pattern covering} technique, introduced by Fomin \emph{et~al.}~\cite{fomin2016subexponential} for planar graphs and later extended to graphs of polynomial growth by Marx and Pilipczuk~\cite{marx2017subexponential}. Formally, we prove that there exists a randomized polynomial-time algorithm which, given an intersection graph of similarly sized objects in $\mathbb{R}^d$ with bounded degree, outputs a vertex set $A \subseteq V(G)$ such that $G[A]$ has treewidth $O(k^{1-1/d}\log k)$, and for every set $X \subseteq V(G)$ of size at most $k$, we have $X \subseteq A$ with probability at least $2^{-\Omega(k^{1-1/d}\log^2 k)}$. This result can be viewed as a strengthening of that of Marx and Pilipczuk~\cite{marx2017subexponential} in the case where the input graph has additional geometric structure rather than merely polynomial growth. Combined together, these techniques yield the first robust subexponential FPT algorithm for \textsc{Long Path}.

\begin{theorem}\label{thm:Main2}
For every constants $d\geqslant 1$ and $\beta\geqslant 1$, there is a robust randomized parameterized algorithm solving \textsc{Long Path} in time $2^{O(k^{1-1/d}\log^2 k)}n^{O(1)}$ on intersection graphs of similarly sized $\beta$-fat objects of $\mathbb{R}^d$.
\end{theorem}

\section{Preliminaries}

\subsection{Graph Theory}
Unless otherwise specified we use standard graph theory terminology. All graphs considered in this paper are simple, undirected, and finite. For a graph $G $, we denote by $V(G)$ and $E(G)$ its vertex set and edge set, respectively. We write $n = |V(G)|$ and $m=|E(G)$| its number of vertices and edges when $G$ is clear from the context. For a vertex $v \in V(G)$, we write $N_G(v)$ for its \emph{open neighborhood}, that is, the set of vertices adjacent to $v$, and $N_G[v] = N_G(v) \cup \{v\}$ for its \emph{closed neighborhood}. For $X \subseteq V$, we denote by $G[X]$ the subgraph of $G$ induced by $X$, and by $G - X$ the subgraph induced by $V \setminus X$. The \emph{degree} of a vertex $v$ is $d_G(v) = |N_G(v)|$, and the maximum degree of $G$ is denoted by $\Delta(G)$. 
An \emph{independent set} $I\subseteq V(G)$ is a subset of pairwise non-adjacent vertices of $G$. The \emph{independence number} of $G$, denoted by $\alpha(G)$, is the maximum size of an independent set of~$G$.

A \emph{path} in a graph $G$ is a sequence of distinct vertices $(v_1, v_2, \ldots, v_k)$ such that any two consecutive vertices are adjacent in $G$. A \emph{cycle} is also a sequence of vertices $(v_1,v_2,\cdots v_k)$ such that any two consecutive vertices are adjacent in $G$, and $v_kv_1\in E(G)$. Given a path $P$ (resp.\ a cycle), we denote by $V(P)$ the set of vertices appearing in the sequence.

For two vertices $u,v$ of a graph $G$, the \emph{distance} between $u$ and $v$, denoted by $\operatorname{dist}_G(u,v)$, is the minimum number of edges on a path from $u$ to $v$ in $G$. For an integer $r$, and a 
vertex $v \in V(G)$, we define $B_G(v,r) = \{\, w \in V(G) : \operatorname{dist}_G(v,w) < r \,\}$,
that is, the set of vertices at distance strictly less than $r$ from $v$ in $G$. 
Similarly, we let $\partial B_G(v,r) = \{\, w \in V(G) : \operatorname{dist}_G(v,w) = r \,\}$
denote the set of vertices at distance exactly $r$ from $v$. 

\subsection{Treewidth, \texorpdfstring{$\kappa$}{kappa}-partition and \texorpdfstring{$\mathcal{P}$}{P}-contraction}
A \emph{tree decomposition} of a graph $G = (V,E)$ is a pair $(T, \{X_t\}_{t \in V(T)})$ where $T$ is a tree, $\{X_t\}_{t \in V(T)}$ is a collection of subsets of $V(G)$ (called \emph{bags}), and satisfying the following:
\begin{enumerate}[label = (\roman*)]
    \item $\bigcup_{t \in V(T)} X_t = V(G)$,
    \item for every edge $uv \in E(G)$, there exists a bag $X_t$ containing both $u$ and $v$, and
    \item for every vertex $v \in V$, the set $\{t \in V(T) \mid v \in X_t\}$ induces a connected subtree of $T$.
\end{enumerate}

Given a graph $G$ and a weight function $\gamma: V(G) \to \mathbb{R}_{\geqslant 0}$, the \emph{weighted width} of a tree decomposition $(T, \{X_t\})$ is 
$
\max_{t\in V(T)} \sum_{v \in X_t} \gamma(v).
$
The \emph{weighted treewidth of $G$} (w.r.t.\ a given weight function) is the minimum weighted treewidth over all tree decompositions of~$G$~\cite{van2007safe}.

 In \cite{de2018framework}, de~Berg \emph{et~al.} introduced the notion of a $\kappa$-partition as a generalization of a clique partition, in order to obtain separator theorems for intersection graphs of similarly sized fat objects in $\mathbb{R}^d$.

\begin{definition}[\texorpdfstring{$\kappa$}{kappa}-partition~\cite{de2018framework}]
Let $G$ be a graph and let $\kappa \geqslant 1$ be an integer.  
A \emph{$\kappa$-partition} of $G$ is a partition $\mathcal{P} = (V_1, \dots, V_t)$ of $V(G)$ such that each class $V_i$ induces a connected subgraph $G[V_i]$ whose vertex set can be partitioned into at most $\kappa$ cliques. If $\kappa = 1$, $\mathcal{P}$ is called a \emph{clique partition} of $G$. 
\end{definition}

\begin{definition}[$\mathcal{P}$-contraction~\cite{de2018framework}]
Given a a graph $G$ and a partition $\mathcal{P} = (V_1, \dots, V_t)$ of $V(G)$, the \emph{$\mathcal{P}$-contraction} of $G$, denoted by $G_\mathcal{P}$, is the (simple) graph whose vertex set is $\{V_1, \dots, V_t\}$, where $V_iV_j$ is an edge if and only if there is an edge in $G$ between a vertex of $V_i$ and a vertex of $V_j$. The \emph{degree} of the partition $\mathcal{P}$ refers to the maximum degree of $G_\mathcal{P}$. Given a weight function $\gamma : \mathbb{R} \to \mathbb{R}$, we define the weighted treewidth of $G_{\mathcal{P}}$ (w.r.t.~$\gamma$) by assigning to each vertex $V_i \in V(G_{\mathcal{P}})$ the weight $\gamma(|V_i|)$.
\end{definition}

The following result shows that for intersection graphs of similarly sized fat objects, it is possible to construct a $\kappa$-partition for which $G_{\mathcal{P}}$ has bounded maximum degree and sublinear weighted treewidth.

\begin{theorem}[de~Berg \emph{et~al.}~{\cite[Theorem~12]{de2018framework}}]
\label{thm:partition-treewidth}
Let $d \geqslant 2$, $\beta \geqslant 1$ and $\varepsilon > 0$ be constants, and let $\gamma$ be a weight function satisfying
$
1 \leqslant \gamma(x) = O(x^{1 - 1/d - \varepsilon}).
$
There exist constants $\kappa$ and $\Delta$ such that the following holds :
\begin{itemize}
    \item Any intersection graph $G$ of similarly-sized $\beta$-fat objects in $\mathbb{R}^d$ has a $\kappa$-partition $\mathcal{P}$  such that $G_\mathcal{P}$ has weighted treewidth $O(n^{1-1/d})$ (with respect to $\gamma$) and maximum degree $\Delta$.
    \item There is a polynomial-time algorithm that, given such a graph $G$ (without the representation), returns such a $\kappa$-partition, and a $2^{O(n^{1-1/d})}$-time algorithm which returns the corresponding tree decomposition.
\end{itemize}
\end{theorem}

\begin{remark}\label{rem:Greedy}
We provide here a few additional observations regarding the construction and properties of the $\kappa$-partition.
\begin{itemize}
    \item The $\kappa$-partition can be obtained in a \emph{greedy manner} as follows: consider a maximal independent set $S$ of $G$ and create one part for each vertex of this independent set.
    By maximality, every vertex $v\in V(G)\setminus S$ has a neighbor in $S$. We pick any $u \in N(v) \cap S$ and add $v$ to the same part as $u$.
    As a consequence, in any geometric representation of $G$ with similarly sized  $\beta$-fat objects, two objects belonging to the same part of a $\kappa$-partition can be assumed to be at distance at most $2\beta$ (see Figure~\ref{fig:fat}).
    \item The graph $G_\mathcal{P}$ defined in the theorem above has (unweighted) treewidth $O(n^{1-1/d})$ when $\gamma$ is the unit function.
\end{itemize}
\end{remark}

\begin{figure}
    \centering
    \begin{tikzpicture}[scale=1.2]
    \tikzset{
        fatobj/.style={fill=gray!40,draw=black,thick},
        pathnode/.style={circle,fill=black,inner sep=1pt},
        diameterarrow/.style={<->,>=latex,thick},
    }

    \coordinate (A) at (0,0);
    \coordinate (B) at (1.8,0);
    \coordinate (C) at (3.6,0);

    \foreach \P in {A,B,C}
        \draw[fill=blue!50,draw=blue, fill opacity = 0.2] (\P) circle (1cm);

    \filldraw[fatobj,rotate=20, fill opacity = 0.4] ($(A)+(0.1,0)$) ellipse (0.8 and 0.9);
    \filldraw[fatobj,rotate=0, fill opacity = 0.4] ($(B)+(0,0)$) ellipse (1 and 0.6);
    \filldraw[fatobj,rotate=-120, fill opacity = 0.4] ($(C)+(0.1,0)$) ellipse (0.8 and 0.9);

    \draw[diameterarrow] ($(B)+(0,-1)$) -- node[right] {$2\beta$} ($(B)+(0,1)$);
\end{tikzpicture}

    \caption{An example of 3 fat objects, the blue disks representing the enclosing disks of each object. The objects on left and right intersect the same object so there are at distance at most $2\beta$.}
    \label{fig:fat}
\end{figure}
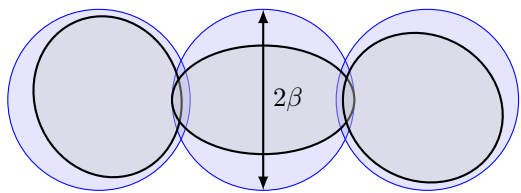

\subsection{Hamiltonicity and Connectivity}

A path (or cycle) is said to be \emph{Hamiltonian} if it contains all vertices of the graph. The \textsc{Hamiltonian Path} (resp.\ \textsc{Cycle}) problem asks, given a graph $G$, whether $G$ has a Hamiltonian path (resp.\ cycle). The \textsc{Long Path} problem asks, given a graph $G$ and a parameter $k$, whether $G$ has a path on $k$ vertices.

Since the problems we study are either trivially infeasible on disconnected graphs, or can be handled by considering each connected component independently, we restrict our attention to connected graphs throughout the paper.

A \emph{vertex separator} in a connected graph $G$ is a subset $A \subseteq V(G)$ such that $G - A = G[V(G) \setminus A]$ is disconnected. The \emph{vertex connectivity} of $G$, denoted by $c_v(G)$, is the size of a minimum vertex separator of $G$, or $|V(G)| - 1$ if $G$ is complete. Since we will not consider edge connectivity, we simply refer to $c_v(G)$ as the \emph{connectivity} of $G$. A graph $G$ is said $\ell$-connected whenever $\ell\leqslant c_v(G)$.

A graph $G$ is said to be \emph{$\ell$-linked} if $|V(G)| \geqslant 2\ell$ and, for every collection of $\ell$ disjoint pairs of distinct vertices $(s_1,t_1), (s_2,t_2), \ldots, (s_\ell,t_\ell)$, there exist $\ell$ vertex-disjoint paths $P_1, \ldots, P_\ell$ such that each $P_i$ connects $s_i$ to $t_i$. The graph $G$ is \emph{Hamiltonian-$\ell$-linked} if there are such path that, in addition, span all the vertices of $G$.

Let $H$ and $G$ be graphs. A \emph{topological minor embedding} (or \emph{TM-embedding}) of $H$ into $G$ is a pair $(M, f)$, where $M$ is a subgraph of $G$ and $f : V(H) \to V(M)$ is an injective mapping such that, for every edge $\{u, v\} \in E(H)$, there exists a path $P_{uv}$ in $M$ connecting $f(u)$ and $f(v)$, and all these paths are pairwise internally vertex-disjoint. We say that $H$ is a \emph{topological minor} of $G$ if such a TM-embedding exists. 
Moreover, when $V(G) = \bigcup_{uv\in E(H)} V(P_{u,v})$, we say that the TM-embedding is \emph{spanning} all vertices of $G$.

The following lemma, due to Fomin \emph{et~al.}~\cite{fomin2024path}, establishes that in highly connected graphs, either a given small graph $H$ can be TM-embedded in a spanning way, or there exists a large independent set.

\begin{lemma}[{\cite[Lemma~4]{fomin2024path}}]\label{lemma:TMembedding}
Let $H$ and $G$ be graphs, with $H$ non-empty. Let $f : V(H) \to V(G)$ be an injective mapping, and let $k$ be a positive integer. Assume that $G$ is $(\max\{k + 2, 10\} \cdot h)$-connected, where $h = |V(H)| + |E(H)|$. Then there exists an algorithm with running time $2^{(h + k)^{O(k)}} + |G|^{O(1)}$ that computes either a subgraph $M \subseteq G$ such that $(M, f)$ is a TM-embedding of $H$ in $G$ spanning all vertices of $G$, or an independent set of size $k$ in $G$.
\end{lemma}

Note that when $H$ is the disjoint union of $\ell$ edges and $k > \alpha(G)$, Lemma~\ref{lemma:TMembedding} implies that the algorithm returns a TM-embedding of $H$ spanning all vertices of $G$, regardless of the choice of the injective mapping $f$. It is straightforward to observe that in this case, $H$ is Hamiltonian-$\ell$-linked.

\begin{lemma}[Corollary of Lemma~\ref{lemma:TMembedding}]\label{lemma:linked}
Let $G$ be a graph and let $\ell \geq 1$ be an integer. If $G$ is $(\max\{\alpha(G) + 3, 10\} \cdot 2\ell)$-connected, then $G$ is Hamiltonian-$\ell$-linked. Moreover, there exists an algorithm with running time $g(\alpha(G), \ell) \cdot |G|^{O(1)}$ (for some computable function $g$) which, given a graph $G$ and $2\ell$ distinct vertices $\{s_i, t_i\}_{1 \leq i \leq \ell}$, constructs $\ell$ vertex-disjoint paths $P_1, \ldots, P_\ell$ such that each $P_i$ connects $s_i$ to $t_i$, and the union of these paths spans all vertices of $G$.
\end{lemma}

\section{A Robust Algorithm for Hamiltonian Path and Cycle}

This section is devoted to our main tool, namely $\lambda$-linked partitions of intersection graphs of similarly sized fat objects, and to their application in proving Theorem~\ref{thm:main1}.

\subsection{Linked Partitions}

We now generalize the notion of $\kappa$-partition to the concept of a \emph{$\lambda$-linked partition}, where each part has to be highly connected.

\begin{definition}[$\lambda$-linked partition]
Let $G$ be a graph and let $\lambda \geqslant 1$ be an integer.  
A \emph{$\lambda$-linked partition} of $G$ is a partition $\mathcal{P} = (V_1, \dots, V_t)$ of $V(G)$ such that each part $V_i$ either induces a clique or a Hamiltonian $\lambda$-linked graph.
\end{definition}


\begin{theorem}\label{thm:ConnectedPartition}
Let $d\geqslant 2$ and $\beta\geqslant 1$ be constants. There exist two constants $\Delta,\lambda$ with $\Delta < \lambda$ such that the following holds:
\begin{itemize}
    \item Any intersection graph $G$ of similarly-sized $\beta$-fat objects in $\mathbb{R}^d$ has a $\lambda$-linked partition $\mathcal{P}$ of $G$ such that $G_{\mathcal{P}}$ has treewidth $O(n^{1-1/d})$ and maximum degree $\Delta$.
    \item There is a $2^{O(n^{1-1/d})}$-time algorithm that, given such a graph (without a representation), returns such a partition as well as the corresponding tree decomposition.
\end{itemize}
Moreover, the difference between $\lambda$ and $\Delta$ can be made arbitrarily large.
\end{theorem}

\begin{proof}
Let $\mathcal{P}_0 = (V_1^0, \dots, V_t^0)$ be the $\kappa$-partition obtained in Theorem~\ref{thm:partition-treewidth}, so that $G_{\mathcal{P}_0}$ has maximum degree $\Delta$ and treewidth $O(n^{1-1/d})$.  
This partition can be computed in polynomial time, and a tree decomposition of width $O(n^{1-1/d})$ in time $2^{O(n^{1-1/d})}$.
We now refine each part $V_i^0$ of $\mathcal{P}_0$ using the following key claim.  
Let $g: \mathbb{N} \to \mathbb{N}$ be a function to be set later.

\begin{claim}\label{claim:partition-gk}
Let $X \subseteq V(G)$ be a vertex set such that $X$ admits a (unknown) clique partition of size $\kappa$.  
Then $X$ can be partitioned into at most $2\kappa$ subsets, each of which either induces a clique or a $g(\kappa)$-connected subgraph, and such a partition can be found in polynomial-time.
\end{claim}

\begin{claimproof}
We first construct a \emph{separator tree} $T$ for $G[X]$. This is a rooted tree defined as follows.  
If $G[X]$ has no vertex separator of size at most $g(\kappa)$, then $T$ consists of a single leaf, which is also the root, and whose label is $X$.  
Otherwise, let $S\subseteq X$ be a vertex separator of $G[X]$ of size at most $g(\kappa)$, and let
$X_{1},\ldots,X_{\ell}$ be the vertex sets of the connected components of $G[X\setminus S]$, where $\ell\ge 2$.  
For each $i\in\{1,\ldots,\ell\}$, recursively construct a separator tree $T_i$ for $G[X_i]$.  
Then $T$ is obtained by creating a new root node labelled with $S$ and making it adjacent to the roots 
of $T_1,\ldots,T_\ell$.

By construction, every leaf of $T$ is labelled with a set $C\subseteq X$ such that $G[C]$ has no separator of size 
at most $g(\kappa)$; that is, each $G[C]$ is $g(\kappa)$-connected.  

Furthermore, $T$ has at most $\kappa$ leaves.  
Indeed, suppose for contradiction that $T$ had at least $\kappa+1$ leaves 
$C_1,\ldots,C_{\kappa+1}$.  
Since the labels of distinct leaves lie in distinct connected components after removing all separators 
encountered on their respective root--leaf paths, they are pairwise non-adjacent.  
Thus picking one vertex from each $C_i$ yields an independent set of size $\kappa+1$, which is a contradiction with the existence of a clique partition of size $\kappa$.

Since every internal node has at least two children and the tree has at most $\kappa$ leaves, 
it follows that $T$ has at most $\kappa-1$ internal nodes.  
Each internal node is labelled with a separator of size at most $g(\kappa)$, so the union $X'$ of those labels satisfies $|X'| \leqslant (\kappa-1) g(\kappa)$.

Finding a vertex separator of size at most $g(\kappa)$ can be done in time $|X|^{g(\kappa)}$.  
Since $T$ contains at most $2\kappa-1$ nodes, this operation is performed at most $2\kappa-1$ times.  
Hence the total running time for constructing $T$ is $(2\kappa - 1)\,|X|^{g(\kappa)}$.
Since $X'\subseteq X$, the set $X'$ admits a partition into at most $\kappa$ cliques.  
As $|X'| \leqslant (\kappa-1)g(\kappa)$, one can find such a partition by brute force in time $\kappa^{(\kappa-1)g(\kappa)}$, by enumerating all partitions of $X'$ into $\kappa$ (possibly empty) parts.
Finally, note that the leaf labels of $T$, together with the clique partition of $X'$, yield a partition of $X$ into at most $2\kappa$ parts, each of them being either a clique or $g(\kappa)$-connected.%
\end{claimproof}

Applying Claim~\ref{claim:partition-gk} to each $V_i^0$ yields a refined partition of $V(G)$ into at most $2\kappa$ subsets per class, each of which induces a clique or a $g(\kappa)$-connected subgraph.  
Let $\mathcal{P}$ denote the resulting partition of $V(G)$, and consider the contraction $G_\mathcal{P}$.  

Since each $V_i^0$ was replaced by at most $2\kappa$ subsets, $G_\mathcal{P}$ has maximum degree at most $\Delta':=(\Delta + 1) \cdot 2\kappa$.  
Moreover, its treewidth remains $O(n^{1-1/d})$, since $\kappa$ is a constant.

Finally, fix $\lambda>\Delta'$ as large as desired (but constant), and fix $g(\kappa)$ large enough such that $g(\kappa)\geqslant \max(\kappa+3,10)\cdot 2\lambda$.  
Then each part of $\mathcal{P}$ is either a clique or Hamiltonian $\lambda$-linked by Lemma~\ref{lemma:linked}, which proves the theorem.

Note that $\mathcal{P}$ can be computed in polynomial time: the partition $\mathcal{P}_{0}$ is obtained in polynomial time, and by Claim~\ref{claim:partition-gk}, the refinement of each $V_i^0 \in \mathcal{P}_0$ can also be computed in polynomial time.

Moreover, a tree decomposition of $G_{\mathcal{P}}$ of treewidth $O(n^{1-1/d})$ can be obtained from a tree decomposition of $G_{\mathcal{P}_0}$ by replacing each vertex $V_i^0$ with the vertices of its refinement in~$\mathcal{P}$.  
Since the treewidth remains $O(n^{1-1/d})$, such a decomposition can be computed in time $2^{O(n^{1-1/d})}$.
\end{proof}

\begin{remark}\label{rem:DistancePartition}
Given a graph $G$, the partition $\mathcal{P}=(V_1,...,V_t)$ of $V(G)$ obtained in Theorem~\ref{thm:ConnectedPartition} is a refinement of the partition $\mathcal{P}_0$ obtained from Theorem~\ref{thm:partition-treewidth}. By Remark~\ref{rem:Greedy}, in any representation $\mathcal{O}=\{O_v\}_{v\in V(G)}$ of $G$ with similarly sized $\beta$-fat objects in $\mathbb{R}^d$, two objects $O_u$ and $O_v$ are at distance at most $2\beta$ whenever $u,v\in V_i$ for some $1\leqslant i\leqslant t$.
\end{remark}

\subsection{The Algorithm}

We begin by adopting the technique introduced by Ito and Kadoshita, who observed that when looking for a Hamiltonian cycle in a graph $G$ admitting a clique partition $\mathcal{P}$ such that the $\mathcal{P}$-contraction has bounded degree, it is sufficient to consider only a constant number of vertices from each clique. In particular, they proved that there always exists a Hamiltonian cycle that uses only a bounded number of edges between any two fixed cliques of the partition.

\begin{lemma}[{\cite[Lemmas~3.2 and 3.3]{ito2010tractability}}]\label{lemma:Ito}
Let $G$ be a graph and $\mathcal{P} = (Q_1, \dots, Q_t)$ a clique partition of $G$ such that $G_\mathcal{P}$ has maximum degree $\Delta$.  
There is a family $\{E_{i,j}\}_{1 \leqslant i < j \leqslant t}$ of edge sets such that the two following points hold:
\begin{enumerate}[label=(\roman*)]
    \item for every $i,j$ such that $1 \leqslant i < j \leqslant t$, $E_{i,j} \subseteq E(G) \cap (Q_i \times Q_j)$ and $|E_{i,j}| \leq 4(2\Delta - 1)^2$;  
    \item $G$ admits a Hamiltonian cycle if and only if it admits one $C$ that, for every $i,j$ as above, uses at most two edges from $(Q_i \times Q_j)$, all belonging to $E_{i,j}$.
\end{enumerate}
Moreover, the sets $\{E_{i,j}\}$ can be computed in polynomial time.
\end{lemma}

\smallskip

In addition, we will use as a black box the FPT algorithms for \textsc{Hamiltonian Path} and \textsc{Hamiltonian Cycle}, both running on a graph $G$ in time $2^{O(\operatorname{tw}(G))} \cdot |G|^{O(1)}$, where $\operatorname{tw}(G)$ denotes the treewidth of $G$.

\begin{theorem}[{\cite{bodlaender2015deterministic, cygan2011solving}}]\label{thm:hamiltonian-tw}
Given a graph $G$ together with a tree decomposition of width $w$, there exists an algorithm that solves \textsc{Hamiltonian Cycle} in time $2^{O(w)} \cdot n^{O(1)}$.
\end{theorem}

\smallskip
We are now ready to prove the main result of this section.

\begin{figure}[t!]
	\centering
	\begin{tikzpicture}[scale=0.55]
		
		\tikzset{
			vertex/.style={circle, fill=black, draw, inner sep=1pt, minimum size=0.2cm},
			rededge/.style={very thick, red},
			blueedge/.style={very thick, blue},
			groupbg/.style={cyan!30, line width=6mm, line cap=round,line join=round},
		}
		
		\foreach \i/\x/\y/\s in {
			1/0/6/$t_1$, 2/0/3/$s_2$, 3/0/0/$t_3$,
		    4/6/6/$s_1$, 5/6/3/$t_2$, 6/6/0/$s_3$,
		    7/3/5/$~$, 8/3/3/$~$, 9/3/1.5/$~$, 10/3/0/$~$
		} \node[vertex, label=above:{\s}] (\i) at (\x,\y) {};
		
		\node[vertex] (11) at (-2,6) {};
		\node[vertex] (12) at (-2,3) {};
		\node[vertex] (13) at (8,3) {};
		\node[vertex] (14) at (8,0) {};
		\node[vertex] (15) at (8,6) {};
		\node[vertex] (16) at (-2,0) {};

		\begin{pgfonlayer}{background}
		      \fill[gray!20, rounded corners=10pt] (-0.8,-1) rectangle (6.8,7);
        
			\draw[groupbg] (15.center) -- (4.center)--(1.center)--(11.center)   -- (12.center) -- (2.center) --(8.center) -- (7.center) -- (5.center) -- (13.center) -- (14.center) -- (6.center) -- (9.center) -- (10.center) -- (3.center) -- (16.center); 

		\end{pgfonlayer}
		
		\draw (1) -- (7) -- (4) ;
		\draw (2) -- (7) -- (5) ;
		\draw (2) -- (9) -- (8) -- (5);
		\draw (3) -- (9) -- (5);
		\draw (3) -- (10) -- (6) ;
		\draw (9) -- (10);
		\draw[rededge] (1)--(4);
		\draw[rededge] (2)--(8) -- (7);
		\draw[rededge] (9)--(6);
		
		\draw[blueedge] (15) -- (4);
		\draw[blueedge] (1) -- (11) ;
		\draw[blueedge] (2) -- (12) ;
		\draw[blueedge] (5) -- (13) ;
		\draw[blueedge] (6) -- (14) ;
		\draw[blueedge] (3) -- (16);
	
	\end{tikzpicture}
	\qquad
		\begin{tikzpicture}[scale=0.55]
		
		\tikzset{
			vertex/.style={circle, fill=black, draw, inner sep=1pt, minimum size=0.2cm},
			rededge/.style={very thick, red},
			blueedge/.style={very thick, blue},
			groupbg/.style={cyan!30, line width=6mm, line cap=round,line join=round},
		}
		
		\foreach \i/\x/\y/\s in {
			1/0/6/$t_1$, 2/0/3/$s_2$, 3/0/0/$t_3$,
		    4/6/6/$s_1$, 5/6/3/$t_2$, 6/6/0/$s_3$,
		    7/3/5/$~$, 8/3/3/$~$, 9/3/1.5/$~$, 10/3/0/$~$
		} \node[vertex, label=above:{\s}] (\i) at (\x,\y) {};
		
		\node[vertex] (11) at (-2,6) {};
		\node[vertex] (12) at (-2,3) {};
		\node[vertex] (13) at (8,3) {};
		\node[vertex] (14) at (8,0) {};
		\node[vertex] (15) at (8,6) {};
		\node[vertex] (16) at (-2,0) {};

		\begin{pgfonlayer}{background}
		      \fill[gray!20, rounded corners=10pt] (-0.8,-1) rectangle (6.8,7);
        
			\draw[groupbg] (15.center) -- (4.center)--(7.center) -- (1.center)--(11.center)   -- (12.center) -- (2.center) --(9.center) -- (8.center) -- (5.center) -- (13.center) -- (14.center) -- (6.center) -- (10.center) -- (3.center) -- (16.center); 

		\end{pgfonlayer}
		
		\draw (1) -- (7) -- (4) ;
		\draw (2) -- (7) -- (5) ;
		\draw (2) -- (9) -- (8) -- (5);
		\draw (3) -- (9) -- (5);
		\draw (3) -- (10) -- (6) ;
		\draw (9) -- (10);

		\draw[blueedge] (15) -- (4);
		\draw[blueedge] (1) -- (11) ;
		\draw[blueedge] (2) -- (12) ;
		\draw[blueedge] (5) -- (13) ;
		\draw[blueedge] (6) -- (14) ;
		\draw[blueedge] (3) -- (16);
	
	\end{tikzpicture}
	\caption{Illustration of the proof of Theorem~\ref{thm:main1}, when turning a Hamiltonian cycle C in H (left) to a Hamiltonian cycle in G (right), focusing on a part $V_i$. In both figures, the grey box represents the set $V_i$, and the blue box represents the (partial) Hamiltonian cycle going through $V_i$. Left: $C$ uses blue edges to enter and leave $V_i$, and uses edges from $G$ as well as red edges inside $V_i$ (notice that we only represented the red edges which are used by $C$). Right: since $G[V_i]$ is Hamiltonian-$\lambda$-linked, the $(s_i, t_i)$-paths which were using red edges in $H$ can be replaced by actual paths in $G$.}
	\label{fig:ReconfigPath}
\end{figure}
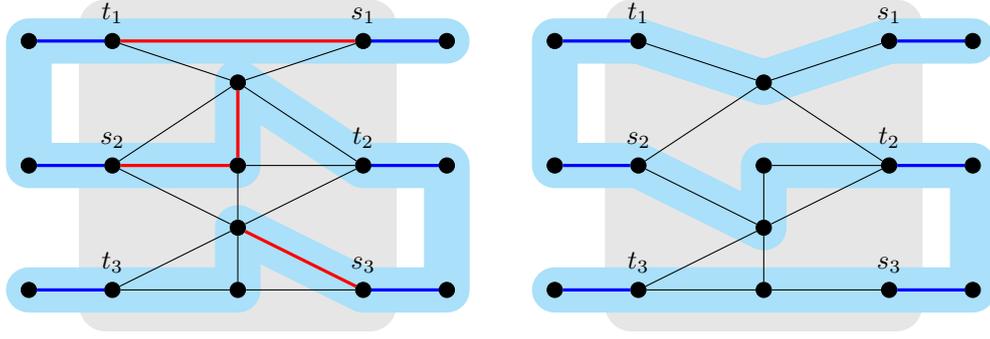

\begin{proof}[Proof of Theorem~\ref{thm:main1}]
We give here an algorithm for \textsc{Hamiltonian Cycle} and explain at the end how to obtain that for \textsc{Hamiltonian Path}.
According to Theorem~\ref{thm:ConnectedPartition}, for some constants $\Delta$ and $\lambda$ with $\Delta<\lambda$ we can compute a $\lambda$-linked partition $\mathcal{P} = (V_1, \dots, V_t)$ of $G$ and a tree decomposition of $G_{\mathcal{P}}$ in time $2^{O(n^{1-1/d})}$.  
In order to be able to apply Lemma~\ref{lemma:Ito}, we define $G^\times$ as the graph obtained by turning each $V_i$ into a clique (i.e., by adding all missing edges inside each $V_i$).  
We call the edges of $E(G^\times) \setminus E(G)$ the \emph{red edges}.

Compute the edge sets $\{E_{i,j}\}_{1 \leqslant i < j \leqslant k}$ from Lemma~\ref{lemma:Ito}, and call these \emph{blue edges}.  
Let $H$ be the graph obtained from $G^\times$ by removing every vertex not incident to a blue edge, except for one special vertex $v_i$ per part $V_i$ (if such a vertex exists).  
Denote by $V_i^H$ the set of vertices of $V_i$ that remain in $H$.
Since $G_\mathcal{P}$ has maximum degree at most $\Delta$,
at most $4\Delta(2\Delta - 1)^2$ vertices from each $V_i$ are incident to blue edges, so $|V_i^H| \leqslant 4\Delta(2\Delta - 1)^2 + 1$.  
Observe that $H$ is a subgraph of the lexicographic product of $G_{\mathcal{P}}$ with a clique of constant size (depending on $d$ only), hence its treewidth is $O(n^{1-1/d})$.

\medskip
We now prove that $G$ admits a Hamiltonian cycle if and only if $H$ does.

\smallskip
Suppose that $G$ has a Hamiltonian cycle. Then so does $G^\times$ as it is a supergraph of $G$.
By Lemma~\ref{lemma:Ito}, $G^\times$ admits a Hamiltonian cycle $C$ that uses at most two blue edges between any $V_i$ and $V_j$, that belong to $E_{i,j}$.
Removing from $C$ any vertex $v \in V_i\setminus V_i^H$ preserves Hamiltonicity in $H$ (as $V_i^H$ is a clique in $H$, and the neighbors of $v$ in $C$ must belong to $V_i$, otherwise $v$ would be adjacent to a blue edge), hence $H$ has a Hamiltonian cycle.

\smallskip
 Conversely, suppose that $H$ has a Hamiltonian cycle $C$.  
We reconstruct a Hamiltonian cycle of $G$ as follows.  
For each $i \in [t]$, if $G[V_i]$ is a clique, we can insert the missing vertices of $V_i \setminus V_i^H$ along $C$. Indeed, if $V_i^H$ contains a special vertex $v_i$, then $C$ must contain an edge $uv_i$ with $u\in V_i$, and thus all vertices of $V_i\setminus V_i^H$ can be inserted between $u$ and $v_i$ in $C$. If $V_i^H$ do not contain a special vertex, then $V_i^H=V_i$ and $C$ already contain all vertices of $V_i$.  

If instead $G[V_i]$ is Hamiltonian-$\lambda$-linked, fix an orientation of $C$ and let $s_1, \dots, s_m$ be the vertices of $V_i$ where $C$ enters $V_i$, and $t_1, \dots, t_m$ where it exits. Observe that $2m\leqslant 2\Delta$ as one can map each $s_j$ (resp.\ $t_j$) to the unique blue edge that $C$ uses to enter (resp.\ exit) $V_i^H$, and $C$ uses at most $2\Delta$ blue edges adjacent to $V_i^H$ (at most two blue edges per adjacent part of $V_i$ in $G_\mathcal{P}$). Since $G[V_i]$ is Hamiltonian-$\lambda$-linked with $\lambda>\Delta$, there exist $m$ vertex-disjoint paths in $G[V_i]$ covering $V_i$ and connecting each $s_j$ to $t_j$.  
Replacing the subpaths of $C$ within $V_i$ by these paths yields a Hamiltonian cycle of $G$. See Figure~\ref{fig:ReconfigPath} for an illustration.

\medskip
To sum up, we reduced the problem to that of deciding whether $H$ has a Hamiltonian cycle. Recall that $H$ has treewidth $O(n^{1-1/d})$.
This can this be done in time $2^{O(n^{1-1/d})}$ using Theorem~\ref{thm:hamiltonian-tw} on the tree decomposition of $H$ computed above. Each local reconstruction inside a part $V_i$ can be carried out using the algorithm of Lemma~\ref{lemma:linked} in polynomial time.

\medskip
With minor adjustments, Lemma~\ref{lemma:Ito} also applies to the \textsc{Hamiltonian Path} problem. More precisely, if $G$ admits a Hamiltonian path, one can construct a graph $G'$ by adding a special edge between the two endpoints of that path. The resulting graph $G'$ admits a clique partition of bounded degree, and thus Lemma~\ref{lemma:Ito} applies to it. The Hamiltonian cycle obtained in $G'$ may include the special edge; removing this edge then yields a Hamiltonian path in the original graph $G$. Then, the very same proof carries over to \textsc{Hamiltonian Path}, yielding an algorithm with identical running time.
\end{proof}

\section{A Subexponential FPT algorithm for \textsc{Long Path}}

This section is dedicated to the proof of Theorem~\ref{thm:Main2}. We first present a low-treewidth pattern covering technique adapted to geometric graphs, and then solve the \textsc{Long Path} problem using both this technique with $\lambda$-linked partitions.

\subsection{Low-Treewidth Pattern Covering}

We say that a graph class $\mathcal{G}$ has \emph{growth} at most $f$ if for every $G\in \mathcal{G}$ and $v\in V(G)$, the ball $B_G(v,r)$ contains at most $f(r)$ vertices.
The next lemma shows that intersection graphs of similarly sized fat objects of bounded maximum degree have growth at most a polynomial function.

\begin{lemma}\label{lemma:growth}
For every constants $d\geqslant 2$, $\beta\geqslant 1$ and $\Delta>0$, any intersection graph of similarly sized $\beta$-fat objects in $\mathbb{R}^d$ has growth $O(r^d)$. 
\end{lemma}

\begin{proof}
Let $\{O_v\}_{v\in V(G)}$ be a representation of $G$ with similarly sized $\beta$-fat objects in $\mathbb{R}^d$. Let $v\in V(G)$.
Observe than since the object associated to every neighbor of $v$ intersects $O_v$, and these objects are $\beta$-fat, there is a $d$-dimensional hypercube of side length $6\beta$ that delimits a region containing the objects of $v$ and all its neighbors. More generally for any integer $r \geqslant 1$, all the vertices at distance at most $r$ from $v$ in $G$ have their objects in the region of $\mathbb{R}^d$ enclosed by a $d$-dimensional hypercube of side length $c_{\beta} r$, for some constant $c_{\beta}$ depending only on $\beta$.  

This hypercube can be partitioned into at most $c_{\beta}^d r^d$ smaller hypercubes of unit side length.  
Since the objects are $\beta$-fat, the object $O_u$ associated to each $u \in V(G)$ contains a ball of radius 1. Let us denote by $z_u$ the center of such a ball. Observe that if a smaller hypercube contains $z_u$ and $z_{u'}$ for some $u,u'\in V(G)$, then $O_u$ and $O_{u'}$ intersect and so $u$ and $u'$ are adjacent.
Hence each such smaller hypercube contains at most $\Delta+1$ points of the form $z_u$ for $u\in V(G)$ (whose associated vertices form a clique of order at most $\Delta+1$).  
Therefore, the ball $B_G(v, r)$ contains at most $c_{\beta}^d (\Delta+1) \cdot r^d$ vertices, as claimed.
\end{proof}

\begin{remark}
As a consequence of the previous lemma and the algorithm of~\cite{marx2017subexponential}, there exists a randomized algorithm that solves \textsc{Long Path} in time $2^{O(n^{1 - 1/(d+1)})}$ in intersection graphs of similarly sized fat objects in $\mathbb{R}^d$, provided that the maximum degree is constant.
\end{remark}

We will use the following lemma showing that in graphs of polynomial growth, one can select a subset of vertices inducing connected components of radius $O(k \log k)$ and that has good probability of containing an unknown subset $X\subseteq V(G)$ of size at most $k$.

\begin{lemma}[{\cite[Lemma~2.1]{marx2017subexponential}}]\label{lemma:poly-growth-cover}
Let $\mathcal{G}$ be a graph class of growth $O(r^\delta)$.
There exists a constant $c > 0$ and a polynomial-time randomized algorithm that, given $G \in \mathcal{G}$ and an integer $k \geqslant 4$, outputs a subset $A \subseteq V(G)$ satisfying:
\begin{enumerate}
    \item every connected component of $G[A]$ has radius at most $c k \log k$;
    \item for every set $X \subseteq V(G)$ of size at most $k$, the probability that $X \subseteq A$ is at least $\tfrac{17}{256}$.
\end{enumerate}
\end{lemma}

We then give a low-treewidth pattern covering lemma for intersection graphs of similarly sized fat objects and bounded maximum degree, with diameter polynomial in $k$.  
The proof, given in Appendix~\ref{sec:pattern}, follows essentially the same clustering procedure as in~\cite{marx2017subexponential}.  
Roughly speaking, the idea is to construct a collection of disjoint \emph{clusters} in the graph, where each cluster induces a subgraph of small treewidth, and then add a small number of vertices that are adjacent to these clusters in order to ensure connectivity between them.  
The main difference with~\cite{marx2017subexponential} lies in the analysis of the treewidth of the clusters, which can be improved in the case of geometric graphs.

\begin{restatable}[\ding{34}]{lemma}{lempat}\label{lemma:pattern}
For every constants $d\geqslant 2$, $\beta\geqslant 1$, $\Delta\geqslant 1$ and $c > 0$, there exists a constant $c' > 0$ and a polynomial-time randomized algorithm that, given $k \geqslant 4$ and a connected intersection graph $G$ of similarly-sized $\beta$-fat objects in $\mathbb{R}^d$ of maximum degree $\Delta$ and radius at most $c k \log k$, outputs a subset $A \subseteq V(G)$ such that:
\begin{enumerate}
    \item the treewidth of $G[A]$ is $O(k^{1 - 1/d} \log k)$;
    \item for every set $X \subseteq V(G)$ of size at most $k$, the probability that $X \subseteq A$ is at least $2^{-c' (|X|+k) k^{-1/d} \log^2 k}$.
\end{enumerate}
\end{restatable}

\subsection{Robust Subexponential FPT algorithm for Long Path}

Recall that there exists an algorithm solving \textsc{Long Path} and \textsc{Long Cycle} in subexponential FPT time $2^{O(\sqrt{k})} \cdot n^{O(1)}$ in unit-disk graphs~\cite{fomin2020eth}, and that this running time is tight under the Exponential Time Hypothesis (ETH). The approach to obtain this algorithm is very similar that used to obtain a subexponential algorithm running in time $2^{O(\sqrt{n})}$. 

The key idea is to view $G$ as a \emph{clique-grid} graph. Informally, a clique-grid graph is a graph equipped with a mapping $f : V(G) \to [t]^2$ such that, for all $(i,j) \in [t]^2$, the preimage $f^{-1}(i,j)$ induces a clique, and edges between vertices in $f^{-1}(i,j)$ and $f^{-1}(i',j')$ exist only if the distance between the cells $(i,j)$ and $(i',j')$ is less than a fixed constant. The algorithm can then be seen as consisting of two main steps: 
\begin{enumerate}[label = (\roman*)]
    \item marking a constant number of vertices in each clique such that the path only traverses the cliques via these marked vertices, and 
    \item applying bidimensionality techniques to solve a weighted version of \textsc{Long Path} on the subgraph induced by the marked vertices in time $2^{O(\sqrt{k})} \cdot n^{O(1)}$.
\end{enumerate}

The mapping $f$ is straightforward to obtain when a geometric representation of the underlying objects is available. However, at the time of writing, it is not know how to do so without the representation. Importantly, the marking scheme does not fundamentally rely on $G$ being a clique-grid graph, but only on the existence of a clique partition $\mathcal{P}$ such that $G_\mathcal{P}$ has bounded degree. The following lemma is a straightforward rewriting of the marking scheme of \cite{fomin2020eth}. It can be viewed as a generalization of Lemma~\ref{lemma:Ito}, which additionally handles the case where the path is not necessarily Hamiltonian.

\begin{lemma}[Reformulation of {\cite[Lemma~14]{fomin2020eth}}]\label{lemma:Mark}
Let $G$ be a graph and $\mathcal{P} = (Q_1, \dots, Q_t)$ a clique partition of $G$ such that the $\mathcal{P}$-contraction $G_\mathcal{P}$ has degree at most $\Delta$.  
There exists a family of vertex subsets $\{\mathsf{M}(i)\}_{1 \leq i \leq t}$ with $\mathsf{M}(i)\subseteq Q_i$ for every $i\in \{1, \dots, t\}$ and satisfying the following property.  
If $G$ contains a path with $k$ vertices, then it also contains a path $P$ with endpoints $x,y$ such that, for every $1 \leq i \leq t$,
\begin{enumerate}
    \item either $V(P) \subseteq Q_i$, or
    \item $V(P) \cap (Q_i \setminus \mathsf{M}(i)) = \emptyset$, or
    \item there exist distinct vertices $u,v \in (V(P) \cap \mathsf{M}(i)) \cup (\{x,y\} \cap Q_i)$  such that the set of internal vertices of the (resp.\ a) subpath of $P$ between $u$ and $v$ is precisely 
    \[
    (V(P) \cap Q_i) \setminus (\mathsf{M}(i) \cup \{u,v\}).
    \]
\end{enumerate}
In addition, $\mathsf{M}$ can be computed in polynomial time, and for any $1\leqslant i \leqslant t$, $|\mathsf{M}(i)| = \Delta^{O(1)}$.
\end{lemma}

In the \textsc{Weighted Long Path} problem, the input is a graph $G$, a weight function $w : V(G) \rightarrow \mathbb{N}$ and an integer $W$, and the goal is to find a path $(u_1, \dots, u_q)$ in $G$ ($q > 0$), such that $\sum_{i \in [q]} w(u_i) \geqslant W$. Adapting the classical dynamic programming algorithm for \textsc{Long Path} on graphs with bounded treewidth~\cite{CyganFKLMPPS15}, we can obtain the following result.

\begin{proposition}[see \cite{CyganFKLMPPS15}]\label{proposition:LongPathTreewidth}
\textsc{Weighted Long Path} is solvable in time $2^{O(w)}\cdot n^{O(1)}$ on $n$-vertex graphs of treewidth at most $w$.
\end{proposition}

We are now ready to prove the main result of this section.

\begin{figure}[h]
\centering

\begin{tikzpicture}[scale=0.8]

\begin{scope}[shift={(0,0)}]
  \node[draw = none] () at (0, 1) {\Large $G$};
  \node[draw, rounded corners, inner sep=5pt, fit={(-2,0.3) (2,-2)}] {};
   \fill[gray!20, rounded corners=10pt] (0.1,-0.1) rectangle (1.9,-1.9);
   \node[draw = none] () at (1,0.1) {$V_j$};
   
   \fill[gray!20, rounded corners=10pt] (-1.9,-0.1) rectangle (-0.1,-1.9);
   \node[draw = none] () at (-1,0.1) {$V_i$};

  \node[draw, circle, fill = blue] (1-1) at (0.5,-0.5) {};
  \node[draw, circle, fill = blue] (1-2) at (1.5,-0.5) {};
  \node[draw, circle] (1-3) at (1.5,-1.5) {};
  \node[draw, circle] (1-4) at (0.5,-1.5) {};
  \draw (1-1) -- (1-2) -- (1-3) -- (1-4) -- (1-1);
  
   \node[draw, circle, fill = blue] (1-5) at (-1.5,-0.5) {};
  \node[draw, circle, fill = blue] (1-6) at (-0.5,-0.5) {};
  \node[draw, circle] (1-7) at (-0.5,-1.5) {};
  \node[draw, circle] (1-8) at (-1.5,-1.5) {};
  \draw (1-5) -- (1-6) -- (1-7) -- (1-8);
  
  \draw (1-1) -- (1-6) --(1-4);
  \draw (1-7) -- (1-4);
\end{scope}

\draw[line width = 3, ->, gray] (2.5,-1) -- (3.5,-1) ;

\begin{scope}[shift={(6,0)}]
  \node[draw = none] () at (0, 1) {\Large $G^\times$};
  \node[draw, rounded corners, inner sep=5pt, fit={(-2,0.3) (2,-2)}] {};
   \fill[gray!20, rounded corners=10pt] (0.1,-0.1) rectangle (1.9,-1.9);
   \node[draw = none] () at (1,0.1) {$V_j$};
   
   \fill[gray!20, rounded corners=10pt] (-1.9,-0.1) rectangle (-0.1,-1.9);
   \node[draw = none] () at (-1,0.1) {$V_i$};

  \node[draw, circle, fill = blue] (1-1) at (0.5,-0.5) {};
  \node[draw, circle, fill = blue] (1-2) at (1.5,-0.5) {};
  \node[draw, circle] (1-3) at (1.5,-1.5) {};
  \node[draw, circle] (1-4) at (0.5,-1.5) {};
  \draw (1-1) -- (1-2) -- (1-3) -- (1-4) -- (1-1);
  \draw[red] (1-1) -- (1-3);
  \draw[red] (1-2) -- (1-4);
  
   \node[draw, circle, fill = blue] (1-5) at (-1.5,-0.5) {};
  \node[draw, circle, fill = blue] (1-6) at (-0.5,-0.5) {};
  \node[draw, circle] (1-7) at (-0.5,-1.5) {};
  \node[draw, circle] (1-8) at (-1.5,-1.5) {};
  \draw (1-5) -- (1-6) -- (1-7) -- (1-8);
  \draw[red] (1-5) -- (1-7) ;
  \draw[red] (1-6) -- (1-8) -- (1-5);
  
  \draw (1-1) -- (1-6) --(1-4);
  \draw (1-7) -- (1-4);
  
\end{scope}


\draw[line width = 3, ->, gray] (6,-2.3) -- (5,-3.3) ;

\begin{scope}[shift={(3,-4)}]
  \node[draw = none] () at (0, 1) {\Large $H$};
  \node[draw, rounded corners, inner sep=5pt, fit={(-2,0.3) (2,-2)}] {};
   \fill[gray!20, rounded corners=10pt] (0.1,-0.1) rectangle (1.9,-1.9);
   \node[draw = none] () at (1,0.1) {$V_j$};
   
   \fill[gray!20, rounded corners=10pt] (-1.9,-0.1) rectangle (-0.1,-1.9);
   \node[draw = none] () at (-1,0.1) {$V_i$};

  \node[draw, circle, fill = blue] (1-1) at (0.5,-0.5) {};
  \node[draw, circle, fill = blue] (1-2) at (1.5,-0.5) {};
  \node[draw, circle] (1-3) at (1,-1.5) {\small $2$};
  \draw (1-1) -- (1-2) -- (1-3) -- (1-1);
  
   \node[draw, circle, fill = blue] (1-5) at (-1.5,-0.5) {};
  \node[draw, circle, fill = blue] (1-6) at (-0.5,-0.5) {};
  \node[draw, circle] (1-7) at (-1,-1.5) {\small $2$};
  \draw (1-5) -- (1-6) -- (1-7) ;
  \draw (1-5) -- (1-7) ;
  
  \draw (1-1) -- (1-6) --(1-3);
  \draw (1-7) -- (1-3);
  
\end{scope}


\draw[line width = 3, ->, gray] (2,-6.3) -- (1,-7.3) ;

\begin{scope}[shift={(0,-8)}]
  \node[draw = none] () at (0, 1) {\Large $H'$};
  \node[draw, rounded corners, inner sep=5pt, fit={(-2,0.3) (2,-2)}] {};
   \fill[gray!20, rounded corners=10pt] (0.1,-0.1) rectangle (1.9,-1.9);
   \node[draw = none] () at (1,0.1) {$V_j$};
   
   \fill[gray!20, rounded corners=10pt] (-1.9,-0.1) rectangle (-0.1,-1.9);
   \node[draw = none] () at (-1,0.1) {$V_i$};

  \node[draw, circle, fill = blue] (1-1) at (0.5,-0.5) {};
  \node[draw, circle, fill = blue] (1-2) at (1.5,-0.5) {};
  \node[draw, circle] (1-3) at (1,-1.5) {\small $2$};
  \draw (1-1) -- (1-2) -- (1-3) -- (1-1);
  
   \node[draw, circle, fill = blue] (1-5) at (-1.5,-0.5) {};
  \node[draw, circle, fill = blue] (1-6) at (-0.5,-0.5) {};
  \node[draw, circle] (1-7) at (-1,-1.5) {\small $2$};
  \draw (1-5) -- (1-6) -- (1-7) ;
  \draw (1-5) -- (1-7) ;
  
  \draw (1-1) -- (1-6) ;
  
\end{scope}


\draw[line width = 3, ->, gray] (4,-6.3) -- (5,-7.3) ;

\begin{scope}[shift={(6,-8)}]
  \node[draw = none] () at (0, 1) {\Large $H^\times$};
  \node[draw, rounded corners, inner sep=5pt, fit={(-2,0.3) (2,-2)}] {};
   \fill[gray!20, rounded corners=10pt] (0.1,-0.1) rectangle (1.9,-1.9);
   \node[draw = none] () at (1,0.1) {$V_j$};
   
   \fill[gray!20, rounded corners=10pt] (-1.9,-0.1) rectangle (-0.1,-1.9);
   \node[draw = none] () at (-1,0.1) {$V_i$};

  \node[draw, circle, fill = blue] (1-1) at (0.5,-0.5) {};
  \node[draw, circle, fill = blue] (1-2) at (1.5,-0.5) {};
  \node[draw, circle] (1-3) at (1,-1.5) {\small $2$};
  \draw (1-1) -- (1-2) -- (1-3) -- (1-1);
  
   \node[draw, circle, fill = blue] (1-5) at (-1.5,-0.5) {};
  \node[draw, circle, fill = blue] (1-6) at (-0.5,-0.5) {};
  \node[draw, circle] (1-7) at (-1,-1.5) {\small $2$};
  \draw (1-5) -- (1-6) -- (1-7) ;
  \draw (1-5) -- (1-7) ;
  
  \foreach \j in {1-5, 1-6, 1-7}{
  		\draw (1-3) -- (\j) ;
  	}
  \foreach \j in {1-1, 1-2}{
  		\draw (1-7) -- (\j) ;
  	}
  \draw (1-1) -- (1-6);
  \draw (1-1) to[out = -160, in = -20] (1-5);
  \draw (1-2) to[out = -160, in = -20] (1-6);
  \draw (1-2) to[out = -160, in = -20] (1-5);

\end{scope}

\end{tikzpicture}

\caption{Illustration of the five graphs of the proof of Theorem~\ref{thm:Main2}.}\label{fig:FiveGraphs}
\end{figure}
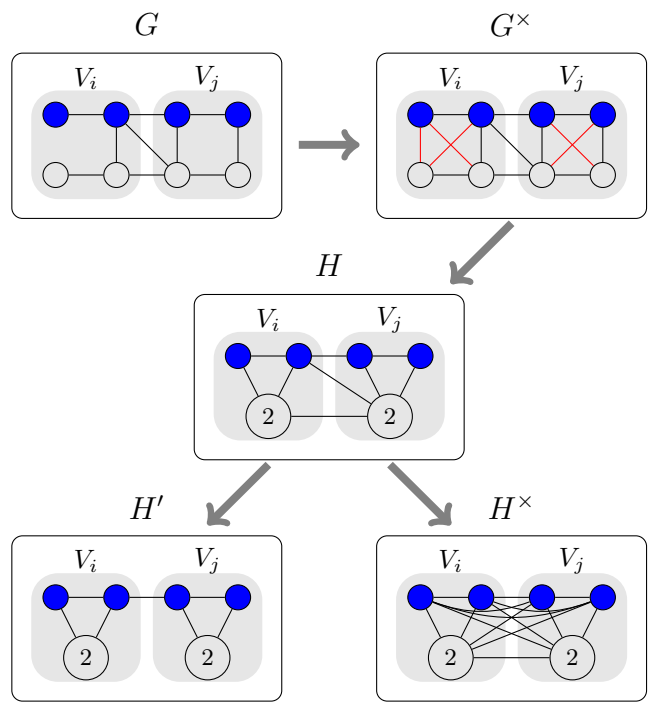

\begin{proof}[Proof of Theorem~\ref{thm:Main2}]
We begin by computing, in polynomial time, the $\lambda$-linked partition 
$\mathcal{P} = (V_1, \dots, V_t)$ guaranteed by Theorem~\ref{thm:ConnectedPartition}, such that $G_{\mathcal{P}}$ has maximum degree $\Delta$, for some constants $\lambda$ and $\Delta$ such that $\Delta < \lambda$.  
Let $G^\times$ denote the graph obtained from $G$ by turning each $V_i$ into a clique, that is, by adding all missing edges inside each $V_i$.  
Again, the edges in $E(G^\times)\setminus E(G)$ are called \emph{red edges}.

We apply Lemma~\ref{lemma:Mark} to the pair $(G^\times, \mathcal{P})$, yielding the marking function 
$\mathsf{M}$ defined on $V(G^\times)$.  
From this, we build the weighted graph $H$ as follows: for each $1\leqslant i\leqslant t$, let 
$U_i := V_i \setminus \mathsf{M}(i)$, and contract all vertices in $U_i$ into a single vertex $v_i$ with weight $w_i := |U_i|$. Whenever $U_i \neq \emptyset$, define $Q_i := \mathsf{M}(i) \cup \{v_i\}$. 
If $U_i = \emptyset$, do not create vertex $v_i$ and define $Q_i := \mathsf{M}(i)$. All remaining vertices (the marked ones) receive weight $1$. Note that, for all $1\leqslant i \leqslant t$, $Q_i$ is a clique in $H$. 

We now define two auxiliary graphs derived from $H$:
\begin{enumerate}
    \item $H'$ is obtained from $H$ by deleting every edge $v_i u$ with $u \in V_j$ and $j\neq i$.  
          
    \item $H^\times$ is obtained from $H$ by adding, for every $1\leqslant i < j \leqslant t$, all edges between $Q_i$ and $Q_j$ whenever there exists an edge $uv \in E(H')$ with $u\in Q_i$ and $v \in Q_j$. Note that, for any $u,v\in Q_i$, $u$ and $v$ are \emph{true twins} in $H^\times$, meaning that $N_{H^\times}[u]=N_{H^\times}[v]$. Moreover, $H'$ is a subgraph of $H^\times$.
\end{enumerate}
See Figure~\ref{fig:FiveGraphs} for an illustration of the graphs $G$, $G^\times$, $H$, $H'$ and $H^\times$. 
We first prove that $H^\times$ is the intersection graph of similarly sized fat objects in $\mathbb{R}^d$.

\begin{claim}
There exist constants $\beta'$ and $\Delta'$ such that $H^\times$ is the an intersection graph of similarly-sized $\beta'$-fat objects in $\mathbb{R}^d$, and has maximum degree at most $\Delta'$.
\end{claim}

\begin{claimproof}
Let $\mathcal{O} = \{O_v\}_{v\in V(G)}$ be a geometric representation of $G$ with $d$-dimensional similarly sized $\beta$-fat objects.  
For each $i$, let $O_i = \bigcup_{v\in V_i} O_v$.  
Since $V_i$ is $\lambda$-linked, the diameter of $O_i$ is bounded by $O(\beta)$ by Remark~\ref{rem:DistancePartition},  
and $O_i$ contains a ball of radius $1$ (as each $O_v$ does).  
We represent every vertex $u \in Q_i$ by the same object $O_i$.  

Note that for each $1 \leqslant i \leqslant t$, the objects representing $Q_i$ pairwise intersect, since they are all represented by the same geometric object. Moreover, for any $1 \leqslant i < j \leqslant t$, the objects $O_i$ and $O_j$ intersect whenever there exist vertices $u \in V_i$ and $v \in V_j$ such that $O_u$ and $O_v$ intersect. It follows that $H^\times$ is indeed the intersection graph of $\{O_u : u\in V(H^\times)\}$. In addition, its degree is bounded by $\Delta' = \Delta^{O(1)}$,  
as each $O_i$ intersects only a bounded number of other regions corresponding to adjacent parts in $G_{\mathcal{P}}$.
\end{claimproof}

We now show that $G$ has a long path if and only if $H'$ does.

\begin{claim}\label{claim:iif}
$G$ has a path on at least $k$ vertices if and only if $H'$ has a weighted path of total weight at least $k$.  
Moreover, given such a path in $H'$, one can construct a path in $G$ with at least the same number of vertices in polynomial time.
\end{claim}

\begin{claimproof}
Suppose that $G$ contains a path $P$ on $k$ vertices with endpoints $x$ and $y$.  
Since $G$ is a subgraph of $G^\times$, $P$ is also a path in $G^\times$.  
By Lemma~\ref{lemma:Mark}, $P$ can be modified into a path $P'$ such that for each $i$,  
$P'$ contains all vertices in $(V(P)\cap V_i)\setminus \mathsf{M}(i)$ consecutively along $P'$, 
and we call $P_i$ such a subpath. By contracting each such $P_i$ into a single vertex $v_i$, we obtain a path $P_H$ in $H'$ whose total weight equals $|V(P)|\geqslant k$.

Conversely, let $P_H$ be a weighted path in $H'$ of total weight at least $k$.  
For each contracted vertex $v_i$ of weight $w_i$, we expand it into an arbitrary path of $w_i$ distinct vertices from $U_i$.  
The two vertices (if they exist) adjacent to $v_i$ in $P_H$ must belong to $\mathsf{M}(i)$ by definition of $H'$, so the resulting expanded path $P'$ exists in $G^{\times}$ and has at least $k$ vertices.  

It remains to transform the path $P'$ in $G^\times$ into a path in $G$.  
Currently, $P'$ may use red edges inside some $V_i$ that is not a clique in $G$, but each $G[V_i]$ is Hamiltonian $\lambda$-linked.  
Fix $i$ such that $G[V_i]$ is Hamiltonian $\lambda$-linked, and let $s_1, \dots, s_m$ be the vertices of $V_i$ where $P'$ enters $V_i$, and $t_1, \dots, t_m$ the vertices where it exits.  
Since $P'$ must enter and exit through marked vertices, we have $m = \Delta^{O(1)}$.  
If $\lambda$ is chosen so that $\lambda > m$, there exist $m$ vertex-disjoint paths $\{P_{i,j}\}_{1\leqslant j \leqslant m}$ covering all vertices of $V_i$, each connecting $s_j$ to $t_j$.  
Replacing each subpath between $s_j$ and $t_j$ in $P'$ by $P_{i,j}$ yields a new path with at least the same number of vertices, and without using red edges.  
Performing this replacement for every part $V_i$ such that $G[V_i]$ is Hamiltonian $\lambda$-linked yields a path on at least $k$ vertices in $G$.  
This construction can be carried out in polynomial time by Lemma~\ref{lemma:linked}.
\end{claimproof}

We now describe the algorithm.  
Suppose that $G$ has a path $P$ with $k$ vertices.  
We apply successively Lemma~\ref{lemma:poly-growth-cover} and Lemma~\ref{lemma:pattern} to $H^\times$ (which, by the previous claim, is the intersection of similarly-size fat objects in $\mathbb{R}^d$, and has bounded maximum degree) to obtain a set $A\subseteq V(H^\times)$ such that:
\begin{enumerate}
    \item the treewidth of $H^\times[A]$ is $O(k^{1-1/d}\log k)$;
    \item with probability at least $2^{-O(k^{1-1/d}\log^2 k)}$, we have $V(P)\subseteq A$.
\end{enumerate}

Since $H'$ is a subgraph of $H^\times$, $H'[A]$ also has treewidth $O(k^{1-1/d}\log k)$.  
Then, one can find a path on $k$ vertices in $H'[A]$ using Proposition~\ref{proposition:LongPathTreewidth} in time $2^{O(k^{1-1/d}\log k)}\cdot n^{O(1)}$.  
Repeating this process $2^{O(k^{1-1/d}\log^2 k)}$ times ensures constant success probability.  
If a path on $k$ vertices is found in $H'$, we construct in polynomial time a path on $k$ vertices of $G$ using Claim~\ref{claim:iif}.  
If no such path is found after all repetitions, the algorithm reports that $G$ has no path on $k$ vertices.
\end{proof}

\section{Conclusion}

In this paper we have presented a robust algorithm for solving \textsc{Hamiltonian Path} and \textsc{Hamiltonian Cycle} in intersection graphs of similarly sized fat objects in $\mathbb{R}^d$ in time $2^{O(n^{1 - 1/d})}$, which is tight under ETH. In addition, we extended our method to obtain a robust randomized FPT algorithm for \textsc{Long Path}, running in time $2^{O(k^{1 - 1/d} \log^2 k)}$. We conclude with two open questions:
\begin{itemize}
    \item Is it possible to remove the $\log^2 k$ factor in the exponent of our algorithm for \textsc{Long Path}, and to make the algorithm deterministic?
    \item Our method does not directly yield a robust subexponential FPT algorithm for \textsc{Long Cycle}. Is it possible to obtain one using a different approach?
\end{itemize}

\bibliographystyle{plain}
\bibliography{biblio}

\appendix

\section{Proof of Lemma~\ref{lemma:pattern}}
\label{sec:pattern}

In this appendix we give the proof of Lemma~\ref{lemma:pattern}, which we will restate below for convenience. We start with a basic observation about the treewidth of intersection graphs of similarly sized fat objects of bounded maximum degree. The statement is an immediate consequence of Theorem~\ref{thm:partition-treewidth}.

\begin{observation}\label{obs:treewidthBoundedDegree}
Let $d\geqslant 1$, $\beta \geqslant 1$, and $\Delta\geqslant 1$ be constants. Any intersection graph $G$ of $n$ similarly-sized $\beta$-fat objects in $\mathbb{R}^d$ that has maximum degree $\Delta$ has treewidth $O(n^{1-1/d})$. In addition, a tree decomposition of width $O(n^{1-1/d})$ can be computed in time $2^{O(n^{1-1/d})}$.
\end{observation}

\begin{proof}
Let $\mathcal{P} = (V_1, \ldots, V_t)$ be the $\kappa$-partition obtained from
Theorem~\ref{thm:partition-treewidth}, together with a tree decomposition
$\mathcal{T} = (T, \{X_t\})$ of $G_{\mathcal{P}}$ of width $O(n^{1-1/d})$.
Since $G$ has maximum degree $\Delta$, each part $V_i$ contains at most
$\kappa \Delta$ vertices of $G$. Recall that since $d$ and $\beta$ are fixed, $\kappa$ is a constant.
Consider the tree decomposition $\mathcal{T}' = (T, \{Y_t\})$ defined by $Y_t := \bigcup_{V_i \in X_t} V_i$ for every $t \in V(T)$.
Then $\mathcal{T}'$ is a tree decomposition of $G$ whose width is at most
$\kappa \Delta$ times the width of $\mathcal{T}$, and therefore still
$O(n^{1-1/d})$.  
Given $\mathcal{T}$, this transformation can be performed in polynomial time.
\end{proof}

\lempat*
\begin{proof}
We use the same random process as in \cite{marx2017subexponential}. 
We initialize $G_0 = G$, $A_0 = \emptyset$, and $B_0 = \emptyset$.  

At iteration $i$, we consider the graph $G_{i-1}$. If $G_{i-1}$ is empty, the process stops. Otherwise, we pick an arbitrary vertex $v_i \in V(G_{i-1})$ and draw a radius $r_i$ from a geometric distribution with success probability $k^{-1/d} \log k$, capped at value $R := c_R k^{1/d}$ (where $c_R$ is a constant yet to be fixed). Concretely, we start with $r_i = 1$, accept this value with probability $k^{-1/d} \log k$, or increase $r_i$ by one and repeat, stopping unconditionally when $r_i = R$. We then set
\[
A_i := A_{i-1} \cup B_{G_{i-1}}(v_i, r_i),
\quad\text{and}\quad
G_i := G_{i-1} \setminus \big( B_{G_{i-1}}(v_i, r_i) \cup \partial B_{G_{i-1}}(v_i, r_i) \big).
\]
With probability $1 - \frac{1}{k|V(G)|}$, we set $P_i = \emptyset$ and $B_i = B_{i-1}$.  
Otherwise, we proceed as follows: we choose uniformly at random an integer
\[
1 \leqslant \ell_i \leqslant k^{1 - 1/d} \log k,
\]
and then select uniformly at random a subset $P_i$ of $\partial B_{G_{i-1}}(v_i, r_i)$ of size $\ell_i$ (or all boundary vertices if fewer than $\ell_i$ exist). We finally set $B_i := B_{i-1} \cup P_i$.

Let $i_0$ be the index of the last iteration. If $|B_{i_0}| > k^{1 - 1/d} \log k$, we output $A = \emptyset$; otherwise, we output $A := A_{i_0} \cup B_{i_0}$.

The following claim is the point where our proof diverges from that of~\cite{marx2017subexponential}.
\medskip
\begin{claim}
The treewidth of $G[A]$ is $O(k^{1-1/d} \log k)$.
\end{claim}
\begin{claimproof}
If $A = \emptyset$ the claim is trivial, so assume otherwise. In particular, $|B_{i_0}| \leqslant k^{1 - 1/d} \log k$.

Note that after deleting the at most $k^{1 - 1/d} \log k$ vertices of $B_{i_0}$, we are left with $G[A_{i_0}]$. Thus, we will construct a tree decomposition of $G[A_{i_0}]$, and then add the vertices of $B_{i_0}$ to every bag.
Let $C\subseteq A_{i_0}$ be a connected component of $G[A_{i_0}]$. Observe that by construction, $C$ is of the form $B_{G_{i-1}} (v_i, r_i)$ for some iteration $i$ of the clustering procedure so in particular $V(C)\subseteq  B_G(v_i,R+1)$. 
By Lemma~\ref{lemma:growth}, there exists a constant $c_g$ such that $|C|\leqslant c_g (R+1)^d = c_g (c_Rk^{1/d}+1)^d=O(k)$. 
Thus, by Observation~\ref{obs:treewidthBoundedDegree}, $G[C]$ has treewidth $O(k^{1-1/d})$.

We can compute a tree decomposition for each connected component, add $B_{i_0}$ to each bag and add a root with bag $B_{i_0}$ connected to the root of every tree decomposition, and we obtain a tree decomposition of $G[A_{i_0}]$ of width $O(k^{1-1/d}\log k)$.
\end{claimproof}

\medskip
\begin{claim}
For every subset $X \subseteq V(G)$ of size at most $k$, the probability that $X \subseteq A$ is at least
$
2^{-c'|X| k^{-1/d} \log^2 k},
$
for some constant $c' > 0$.
\end{claim}

\begin{claimproof}
Fix $X \subseteq V(G)$. The claim is trivial for $X = \emptyset$, so assume $|X| \geqslant 1$.

For an iteration $i$, given $v_i$ and $G_{i-1}$, we call a radius $r\in \{1,\ldots R\}$ \emph{bad} if
\begin{equation}\label{eq:BadRadius}
|X \cap \partial B_{G_{i-1}}(v_i, r)| >
k^{-1/d} \log k \cdot |X \cap B_{G_{i-1}}(v_i, r)|.
\end{equation}
Let $1 \leqslant \rho_0 < \rho_1 < \dots < \rho_t$ denote all bad radii for iteration $i$.
We first bound $t$, the number of bad radii in a fixed iteration.  
For any $j \geqslant 1$, we have
\[
\partial B_{G_{i-1}}(v_i, \rho_j) = B_{G_{i-1}}(v_i, \rho_j + 1)
\setminus B_{G_{i-1}}(v_i, \rho_j) \subseteq B_{G_{i-1}}(v_i, \rho_{j+1}).
\]
It follows that
\begin{align*}
|X \cap B_{G_{i-1}}(v_i, \rho_{j+1})|
&\geqslant |X \cap B_{G_{i-1}}(v_i, \rho_{j})| + |X \cap \partial B_{G_{i-1}}(v_i, \rho_{j})| \\
&\geqslant \left( 1 + k^{-1/d} \log k \right) |X \cap B_{G_{i-1}}(v_i, \rho_{j})|,
\end{align*}
where the last inequality follows from Equation~\eqref{eq:BadRadius}.  
Since $|X \cap B_{G_{i-1}}(v_i, \rho_1)| \geqslant 1$ and $|X \cap B_{G_{i-1}}(v_i, \rho_{t})| \leqslant k$, we obtain
\begin{equation}\label{eq:plus1}
k \geqslant \left(1 + k^{-1/d} \log k\right)^{t-1}.
\end{equation}
Note that there is a positive constant $c_1 < 1$, depending only on $d$, such that for all $k \geqslant 1$,
\[
c_1 \, k^{-1/d} \log k \leqslant \log\left(1 + k^{-1/d} \log k\right).
\]
Taking logarithm in \eqref{eq:plus1} and using the above, we get
\begin{equation}\label{eq:Boundt}
t < c_1^{-1} k^{1/d}.
\end{equation}

\smallskip
Let $\mathbf{A}$ denote the event that (i) every chosen radius $r_i$ is not bad, and (ii) $r_i < R$.  
Whenever a bad radius $r$ is skipped, at least one vertex from $X \cap \partial B_{G_{i-1}}(v_i, r)$ joins $A_i$; hence there are at most $|X|$ bad radii in total.  
Each such radius is not accepted independently with probability $1 - k^{-1/d} \log k$.  
Therefore, the probability of (i) is at least
\[
(1 - k^{-1/d} \log k)^{|X|}.
\]

By Equation~\eqref{eq:Boundt}, there are at most $c_1^{-1} k^{1/d}$ bad radii in iteration $i$.  
Thus, the probability that $r_i$ reaches the cap $R$ without accepting a good radius is at most
\[
(1 - k^{-1/d} \log k)^{R - c_1^{-1} k^{1/d}}
= e^{(c_R - c_1^{-1}) k^{1/d} \log(1 - k^{-1/d} \log k)}
\leqslant e^{-(c_R - c_1^{-1}) \log k},
\]
where the last inequality holds since $\log x \leqslant x - 1$ for all $x > 0$.  
By the union bound and Lemma~\ref{lemma:growth}, the probability that some $r_i$ reaches the cap $R$ is at most
\[
|V(G)| \, k^{-(c_R - c_1^{-1})}
\leqslant c_{g} (c k \log k)^d \, k^{-(c_R - c_1^{-1})}.
\]
Choosing $c_R$ sufficiently large ensures this probability is at most $k^{-1}$.  
Hence, the event $\mathbf{A}$ occurs with probability at least
\[
(1 - k^{-1}) (1 - k^{-1/d} \log k)^{|X|}
\geqslant 2^{-c_{\mathbf{A}} |X| k^{-1/d} \log k}
\]
for some constant $c_{\mathbf{A}} > 0$.

\smallskip
Assume now that $\mathbf{A}$ occurs and fix the realizations of all $v_i$ and $r_i$.  
Then the sets $A_i$ and graphs $G_i$ are determined, and only the choices of $P_i \subseteq \partial B_{G_{i-1}}(v_i, r_i)$ remain random.  
For each iteration $i$, define $X_i := X \cap \partial B_{G_{i-1}}(v_i, r_i)$.  
We consider the event $\mathbf{B}$ that $P_i = X_i$ for all $i$.  
If $\mathbf{B}$ holds, then $X \subseteq A$.

Whenever $X_i = \emptyset$, the algorithm may incorrectly set $P_i \ne \emptyset$ with probability $1/(k|V(G)|)$.  
By the union bound, the probability that such an error occurs for at least one iteration is at most $k^{-1}$.

Assuming all these empty cases are handled correctly, consider an iteration $i$ with $X_i \ne \emptyset$.  
Since $r_i$ is good, we have
\begin{equation}\label{eq:BoundXi}
|X \cap \partial B_{G_{i-1}}(v_i, r_i)|
\leqslant k^{-1/d} \log k \cdot |X \cap B_{G_{i-1}}(v_i, r_i)|.
\end{equation}
It follows that $|X \cap B_{G_{i-1}}(v_i, r_i)| \geqslant \frac{k^{1/d}}{\log k}$.  
Hence, there can be at most $k^{1 - 1/d} \log k$ iterations $i$ such that $X_i \ne \emptyset$.  
Moreover,
\begin{equation}\label{eq:BoundSumXi}
\sum_{i : X_i \ne \emptyset} |X_i|
\leqslant k^{-1/d} \log k \sum_{i : X_i \ne \emptyset} |X \cap B_{G_{i-1}}(v_i, r_i)|
\leqslant |X| k^{-1/d} \log k.
\end{equation}

We are now ready to bound the probability of event $\mathbf{B}$.  
In each iteration with $X_i \ne \emptyset$, the algorithm must independently:
\begin{enumerate}[label=(\roman*)]
    \item guess that $P_i$ is nonempty (probability $1/(k|V(G)|)$);
    \item guess $\ell_i = |X_i|$ (probability at least $1/(k^{1 - 1/d} \log k)$);
    \item guess $P_i = X_i$ (probability at least $|V(G)|^{-|X_i|}$).
\end{enumerate}
The bound in (ii) follows since $\ell_i$ is chosen uniformly at random in $\{1, \ldots, k^{1 - 1/d} \log k\}$, and $|X_i| \leqslant k^{1 - 1/d} \log k$ by combining Equation~\eqref{eq:BoundXi} with $|X \cap B_{G_{i-1}}(v_i, r_i)| \leqslant k$.

Hence, the probability that all $X_i \ne \emptyset$ are guessed correctly is at least
\begin{align*}
&\prod_{i : X_i \ne \emptyset}
\frac{1}{k|V(G)|} \cdot
\frac{1}{k^{1 - 1/d} \log k} \cdot
\frac{1}{|V(G)|^{|X_i|}} \\
\geqslant &
\left(\frac{1}{k|V(G)|} \cdot \frac{1}{k^{1 - 1/d} \log k}\right)^{k^{1 - 1/d} \log k}
\left(\frac{1}{|V(G)|}\right)^{|X| k^{-1/d} \log k},
\end{align*}
where the inequality follows from Equation~\eqref{eq:BoundSumXi}.  
Let $c_2 > 0$ be a sufficiently large constant such that both $k |V(G)| k^{1 - 1/d} \log k \leqslant k^{c_2}$ and $|V(G)| \leqslant k^{c_2}$.
Then the probability of $\mathbf{B}$ is at least
\[
(1 - 1/k) \cdot k^{-c_2 (|X| + k) k^{-1/d} \log k}
\geqslant 2^{-c_{\mathbf{B}} (|X| + k) k^{-1/d} \log^2 k}
\]
for some constant $c_{\mathbf{B}} > 0$.  
This concludes the proof of the claim.
\end{claimproof}

\smallskip
Combining the two claims yields the desired lemma.
\end{proof}

\end{document}